\documentclass[final, nomarks
final
]{dmtcs-episciences}

\usepackage[utf8]{inputenc}
\usepackage{subfigure}

\usepackage{float,cite}
\usepackage{todonotes}
\usepackage[basic]{complexity}
\usepackage{thmtools}
\usepackage{amsmath,amssymb,bm,amsthm}
\usepackage{etoolbox}
\usepackage{hyperref}
\usepackage{multirow}
\usepackage{booktabs}
\usepackage[capitalise]{cleveref}

\newcommand{\RR}{\ensuremath{\mathbb{R}}}
\newtheorem{lemma}{Lemma}
\newtheorem{corollary}{Corollary}
\newtheorem{theorem}{Theorem}
\newtheorem{definition}{Definition}
\newtheorem{observation}{Observation}

\newtheorem{problem}{Problem}
\newcommand{\revision}[1]{#1}
\usepackage[round]{natbib}

\author[Dobler, Kobourov, Mondal, and Nöllenburg]{Alexander Dobler\affiliationmark{1}\thanks{Supported by the Vienna Science and Technology Fund (WWTF) under grant [10.47379/ICT19035].}
  \and Stephen Kobourov\affiliationmark{2}\\
  \and Debajyoti Mondal\affiliationmark{3}\thanks{Supported by the Natural Sciences and Engineering Research Council of Canada (NSERC).}
  \and Martin Nöllenburg\affiliationmark{1}\footnotemark[1] %
  }

\title[Representing Hypergraphs by Point-Line Incidences]{Representing Hypergraphs by Point-Line Incidences}
\affiliation{
  Algorithms and Complexity Group, TU Wien, Vienna, Austria\\
  Department of Computer Science, Technical University of Munich, Munich, Germany\\
  Department of Computer Science, University of Saskatchewan, Saskatoon, Canada}
\begin{document}
\publicationdata{vol. 28:3}{2026}{3}{10.46298/dmtcs.15876}{2025-06-16; 2025-06-16; 2026-06-09}{2026-06-25}
\maketitle
\begin{abstract}
    We consider hypergraph visualizations that represent vertices as points in the plane and hyperedges as curves passing through the points of their incident vertices. Specifically, we consider several different variants of this problem by (a) restricting the curves to be lines or line segments, (b) allowing two curves to cross if they do not share an element, or not; and (c) allowing two curves to overlap, or not.
    We show $\exists\RR$-hardness for six of the eight resulting decision problem variants
    and describe polynomial-time 
    algorithms in some restricted settings.
    Lastly, we briefly touch on what happens if we allow the lines of the represented hyperedges to have bends---to this we generalize a counterexample to a long-standing result that was sometimes assumed to be correct.
\end{abstract}

\section{Introduction}
Hypergraphs, or equivalently set systems,  arise in many domains and visualizing them is a non-trivial challenge. Classical approaches, such as Venn and Euler diagrams~\citep{jp-hpcdvd-87,amahmr-sv-16,m-dh-90} do not scale to large instances. Recent experimental work~\citep{wallinger2021readability} has shown that representing hypergraphs as collections of polylines (for the hyperedges) and common intersection points (for the vertices) allows for faster and more accurate performance of hypergraph-related tasks.
\begin{figure}
    \centering
    \includegraphics[width=\textwidth]{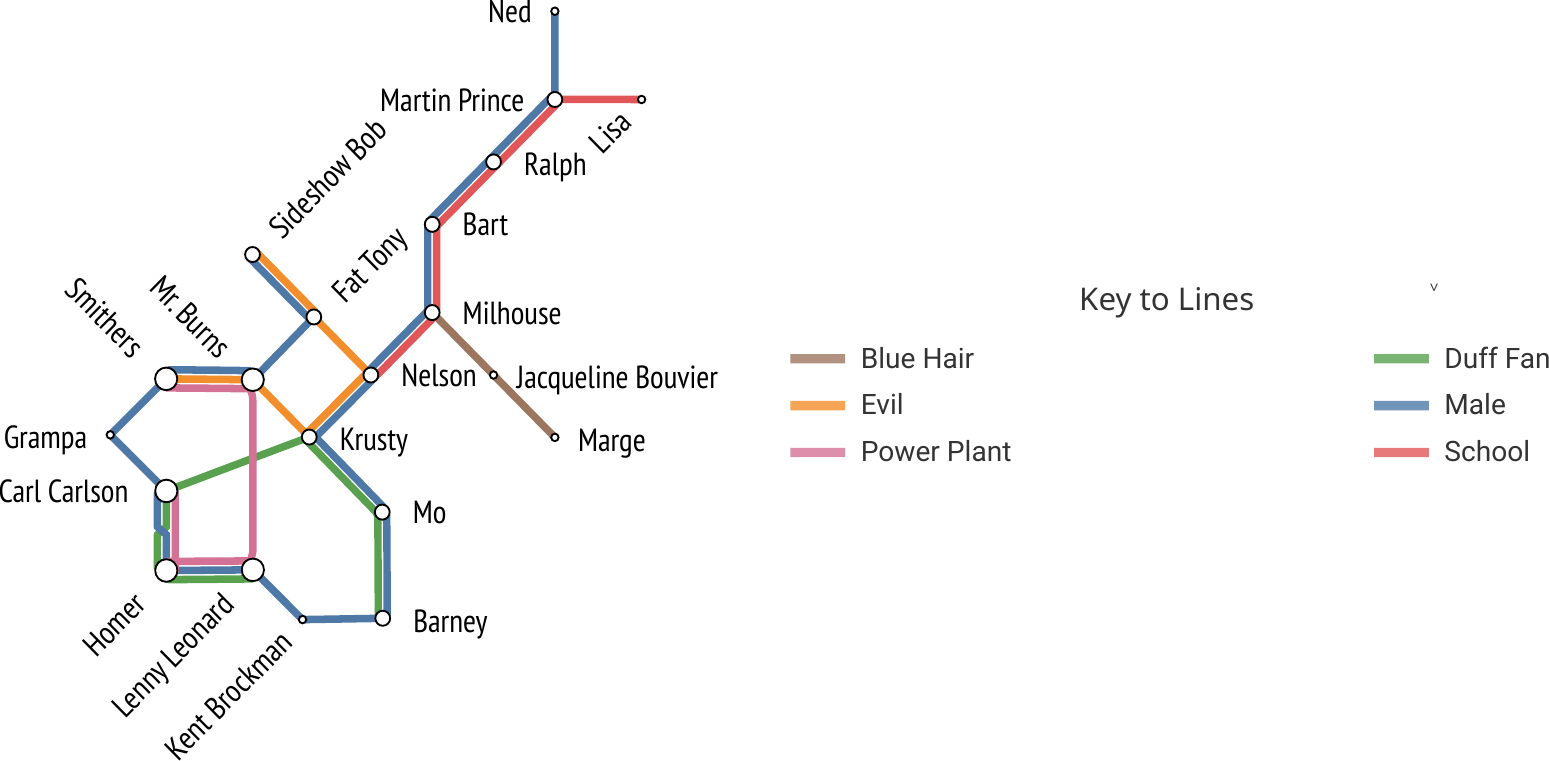}
    \caption{A visualization of a Simpsons hypergraph dataset using the MetroSets metaphor \citep{jacobsen2020metrosets}. The visualization is taken from \url{https://metrosets.ac.tuwien.ac.at/}. Hyperedges are represented by metro lines and elements are represented by stations.}
    \label{fig:metrosets}
\end{figure}
In particular, the LineSets~\citep{ahrc-dslnvt-11} and the MetroSets approach~\citep{jacobsen2020metrosets} use the metro map metaphor, where each hyperedge is a metro line and each vertex an interchange station.
While LineSets connects pre-embedded vertices with arbitary curves, MetroSets optimizes vertex positions and aims to visualize the result in an octolinear style; see \cref{fig:metrosets}.
Minimizing the visual complexity (which depends on the total number of bends along the metro lines),
makes the representations simpler to understand and work with.
A natural question is: which hypergraphs can be represented with just one bendless line segment per hyperedge?

With this in mind, we study the
problem of representing a hypergraph with vertices as points in the plane and hyperedges as straight lines through incident vertices.
Specifically, we consider several problems by varying one of the following requirements (see also \cref{fig:representations} and formal definitions in \cref{section:preliminaries}):
\begin{enumerate}
    \item[(a)] The curves must be line segments or (infinite) lines. %
    \item[(b)] Two curves are allowed to cross if they do not share an element, or not.
    \item[(c)] Two curves are allowed to overlap (i.e., to share a line segment), or not (called \emph{strict}).
\end{enumerate}

\begin{table}[b]
    \centering
    \caption{Complexity results for deciding whether a representation that possibly has crossings exists. The \emph{rank} of a hypergraph is the cardinality of its largest hyperedge. The value n.b.\ means unbounded max-degree or rank.}
    \label{table:resultswcrossings}
    \setlength{\tabcolsep}{3pt}
    \begin{tabular}{lccccccccc}
        \toprule
                                           & \multicolumn{3}{c}{lines}                 & \multicolumn{6}{c}{line segments}                                                                                                                                                                                                               \\
        \cmidrule(r){2-4}  \cmidrule(l){5-10}
                                           & \multicolumn{3}{c}{strict (= non-strict)} & \multicolumn{3}{c}{strict}        & \multicolumn{3}{c}{non-strict}                                                                                                                                                                              \\
        \cmidrule(r){2-4}  \cmidrule(lr){5-7}  \cmidrule(l){8-10}
        complexity                         & rank                                      & max-deg                           & ref                            & rank                     & max-deg                   & ref                                                 & rank    & max-deg & ref                                       \\
        \midrule
        \multirow{2}{*}{$\exists\RR$-hard} & $\ge 3$                                   & n.b.                              & Thm \ref{thm:erlinesrank3}     & \multirow{2}{*}{$\ge 3$} & \multirow{2}{*}{$\ge 16$} & \multirow{2}{*}{Thm \ref{thm:erstrictlinesegments}} & $\ge 3$ & $\ge 6$ & Thm \ref{thm:erlinesgmentsrank3}          \\
                                           & n.b.                                      & $\ge 3$                           & Cor \ref{thm:erlinesdeg3}      &                          &                           &                                                     & $\ge 5$ & $\ge 2$ & Thm \ref{thm:erlinesgmentsmdeg2}          \\
        \midrule
        \multirow{2}{*}{Poly-time}         & $\le 2$                                   & n.b.                              & Obs \ref{obs:linerank-2}       & $\le 2$                  & n.b.                      & Cor \ref{corr:strictlinesegmentpoly}                & $\le 2$ & n.b.    & Cor \ref{corr:strictlinesegmentpoly}      \\
                                           & n.b.                                      & $\le 2$                           & Obs \ref{obs:linerank-2}       & n.b.                     & $\le 2$                   & Cor \ref{corr:strictlinesegmentpoly}                & $\le 3$ & $\le 2$ & Thm \ref{thm:polylinesegmentrank3degree2} \\
        \bottomrule
    \end{tabular}
\end{table}
\begin{table}[t]
    \centering
    \caption{Complexity results for deciding whether a crossing-free representation exists. The value n.b.\ means unbounded max-degree or rank.}
    \label{table:resultswocrossings}
    \setlength{\tabcolsep}{3pt}
    \begin{tabular}{lccccccccc}
        \toprule
                                           & \multicolumn{3}{c}{lines}                 & \multicolumn{6}{c}{line segments}                                                                                                                                                                             \\
        \cmidrule(r){2-4}  \cmidrule(l){5-10}
                                           & \multicolumn{3}{c}{strict (= non-strict)} & \multicolumn{3}{c}{strict}        & \multicolumn{3}{c}{non-strict}                                                                                                                                            \\
        \cmidrule(r){2-4}  \cmidrule(lr){5-7}  \cmidrule(l){8-10}
        complexity                         & rank                                      & max-deg                           & ref                              & rank    & max-deg  & ref                                           & rank    & max-deg & ref                                           \\
        \midrule

        \multirow{2}{*}{$\exists\RR$-hard} & \multirow{2}{*}{?}                        & \multirow{2}{*}{?}                & \multirow{2}{*}{--}              & $\ge 5$ & $\ge 12$ & Thm \ref{thm:ercrossingfreestrictlinesegment} & $\ge 3$ & $\ge 6$ & Cor \ref{cor:ercrossingfreelinesegmentsrank3} \\
                                           &                                           &                                   &                                  & n.b.    & $\ge 2$  & Thm \ref{thm:erstrictlinesegmentmdeg2}        & $\ge 5$ & $\ge 2$ & Cor \ref{cor:ercrossingfreelinesegmentsmdeg2} \\

        \midrule

        \multirow{2}{*}{Poly-time}         & $\le 2$                                   & n.b.                              & Thm \ref{thm:crossfreelines}     & $\le 2$ & n.b.     & Obs \ref{obs:segmentrank2planarity}           & $\le 2$ & n.b.    & Obs \ref{obs:segmentrank2planarity}           \\
                                           & n.b.                                      & $\le 2$                           & Thm \ref{thm:crossfreelinesdeg2}                                                                                                                                          \\

        \bottomrule
    \end{tabular}
\end{table}

\subparagraph*{Contributions.} In an extensive complexity study, we investigate for which hyperedge cardinalities and vertex degrees each of the eight problems is $\exists\RR$-hard or solvable in polynomial time.
We also study special graph classes that always admit such representations.
Our contributions %
are detailed in \cref{table:resultswcrossings,table:resultswocrossings}.
Our results are structured by distinguishing representations, where crossings are permitted (\cref{section:representations}, \cref{table:resultswcrossings}) or not permitted (\cref{section:crossfreerepresentations}, \cref{table:resultswocrossings}).
\cref{section:crossfreerepresentations} further presents two hypergraph classes that always admit a crossing-free representation with segments.
\cref{section:beyondzerobend} discusses a century-old claim by Steinitz, and generalizes a counterexample. We start with preliminaries (\cref{section:preliminaries}) and a description of the Pappus configuration (\cref{section:pappus}) as an essential tool for constructing gadgets in our reductions.

\begin{figure}[bt]
    \centering
    \includegraphics[width=\textwidth]{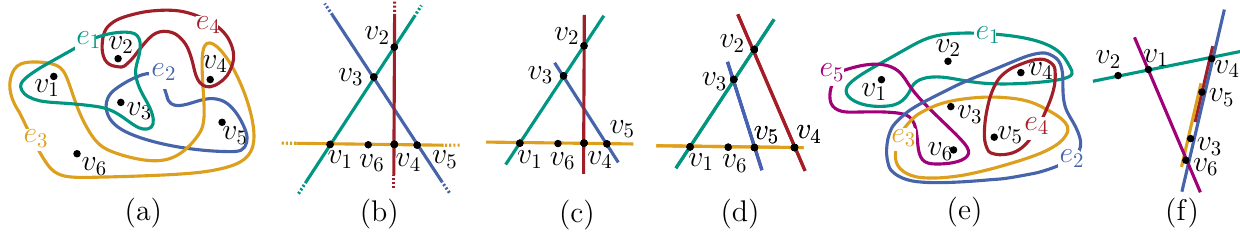}
    \caption{(a) A linear hypergraph $H$. (b) A strict line representation of $H$. (c) A strict segment representation of $H$. (d) A crossing-free strict segment representation of $H$. (e)--(f) A hypergraph and its non-strict segment representation.}
    \label{fig:representations}
\end{figure}

\subparagraph*{Related Work.}
Representing hypergraphs with one line per hyperedge relates to classical geometric problems dating back to the 19th century, particularly in the study of configurations~\citep{DBLP:books/ph/Gruenbaum09}. A {\em combinatorial} (or {\em geometric}) {\em configuration} is an abstract (or Euclidean) incidence structure with points and lines where each point is incident to the same number of lines, and each line to the same number of points. A {\em realization} of a combinatorial configuration is a geometric configuration, i.e., an embedding of the points and lines in the Euclidean plane that contains exactly the incidences that appear in the configuration. This was central to Steinitz's PhD thesis~\citep{steinitz1894konstruction} and studied in notable configurations like those of Desargues, Pappus, and Möbius–Kantor~\citep{DBLP:books/ph/Gruenbaum09}.
\Citet{steinitz1894konstruction} claimed that in a 3-uniform 3-regular linear hypergraph (where each point lies on three lines, each line passes through three points, and no two lines share more than one point), removing one point allows for a realization. This claim was later proven incorrect~\citep{DBLP:books/ph/Gruenbaum09,flowers2015embeddings}.
Citet{g-dc-95} had already noted the connection between configurations and hypergraph drawings, arguing that using straight lines for hyperedges improves readability. He later collected classical results and posed hypergraph realizability as an open question~\citep{DBLP:journals/dm/Gropp97}.

Deciding whether the representation of a hypergraph can be a crossing-free straight-line drawing of a tree can be done in polynomial time~\citep{DBLP:journals/siamcomp/SwaminathanW94}. A \emph{support graph} $G$ of a hypergraph $H$ is a graph with the same vertex set where each hyperedge of $H$ is connected in $G$.
Techniques for drawing hypergraphs via support graphs~\citep{bkmsv-psh-11a,BuchinKMSV09,bcps-psh-12,cgmny-spssh-19,cgmny-spssh-18} have a different focus and do not take into account whether vertices along hyperedges are collinear. Minimizing the number of crossings in a path-based support graph (hyperedges form a path in $G$) is \NP-hard~\citep{DBLP:conf/sofsem/FrankKKMPUW21}.
\Citet{DBLP:journals/dmtcs/FirmanS25}  consider special representations of $k$-uniform hypergraphs, where each vertex is an axis-aligned affine $\ell$-dimensional plane in $\RR^d$ ($d>\ell$ and $k=\binom{d}{\ell}$) and each hyperedge is a point in $\RR^d$ that is part of the planes defined by its incident vertices. They give a characterization by vertex cuts and give a polynomial-time algorithm to test whether a hypergraph has such a representation.
\Citet*{DBLP:conf/gd/BertschingerMKMW23} consider the problem of representing hypergraphs such that vertices are points in $\RR^d$ and hyperedges are bi-curved, difference-separable, convex sets cointaining their incident points. They show $\exists\RR$-hardness of the associated recognition problems.

The intersection graph of line segments~\citep{DBLP:journals/corr/Matousek14} can be seen as a strict segment representation of a linear hypergraph.
\Citet{DBLP:journals/ejc/Goncalves09} showed that some planar linear hypergraphs (see \citep{DBLP:journals/ejc/Goncalves09} for a definition) cannot be represented with straight line segments (in contrast to planar graphs).
Segment contact representations of planar graphs can sometimes help in finding crossing-free strict line representations for linear hypergraphs.
A necessary and sufficient condition for representing a graph as a contact system of segments is known~\citep{de2007representation,DBLP:journals/algorithmica/FraysseixM07}, but no polynomial-time algorithm exists to test it.
The hypergraph visualization problem is also similar to the stretchability problem~\citep{shor1991stretchability,s-csgtp-09}:  given an arrangement of pseudolines, is there a combinatorially equivalent arrangement of lines?
Unlike hypergraph representations, the order of vertices along each pseudoline is fixed in the stretchability problem.
Matroid representability~\citep{DBLP:journals/corr/abs-2301-03221} is another related problem, asking whether elements of a matroid can be represented as vectors in~$\RR^3$ such that independent sets are retained, becoming similar to hypergraph representations through projective transformation onto the plane.

\section{Preliminaries}\label{section:preliminaries}
A hypergraph $H = (V,E)$ is defined by a vertex set $V$ and a hyperedge set $E$, \revision{where each $e\in E$ is a distinct non-empty subset of $V$.}
The \emph{degree} of a vertex $v$ in $H$ is the number of hyperedges containing~$v$. The \emph{rank} of $H$ is the maximum cardinality $|e|$ over all hyperedges $e$ in $E$.
A hypergraph $H$ is \emph{$k$-uniform} if every hyperedge has cardinality $k$ and it is \emph{$k$-regular} if every vertex has degree~$k$. It is \emph{linear} if $|e\cap e'|\le 1$ for every pair of distinct hyperedges $e,e'\in E$.
The \emph{dual hypergraph} of $H$ is obtained by interchanging the role of vertices and hyperedges, i.e., the dual of $H=(V,E)$ is $H_d=(E,E')$ where $E'$ consists of hyperedges $\{e\mid v\in e\}$ for each $v\in V$. The \emph{hyperedge intersection graph} $G$ of $H$ is a graph with $E$ as its vertex set, where two vertices are adjacent if and only if the corresponding hyperedges share a common element.

A point $x$ is \emph{incident} to a line/line segment $\ell$ if and only if $\ell$ contains $x$.
A \emph{line representation} of a hypergraph consists of an injective mapping $\alpha$ of vertices to points in $\RR^2$ and an injective mapping $\beta$ of hyperedges to lines in $\RR^2$ such that $v\in e$ if and only if $\alpha(v)\in \beta(e)$ for $v\in V,e\in E$.
A \emph{segment representation} of a hypergraph is defined as a line representation with line segments instead of lines.
A line representation/segment representation is \emph{strict} if every pair of lines/line segments share at most one point.
It is crossing-free if a pair of lines/line segments $\beta(e)$ and $\beta(e')$ in the representation share a point if and only if $e\cap e'\ne \emptyset$.
We are concerned with the problem of deciding whether a hypergraph has a representation. We ask the following question.
\begin{problem}[$H$-Representation]
Given a hypergraph $H$, $\mathcal{S}\in\{\text{strict, non-strict}\}$, $\mathcal{C}\in\{\text{non-crossing-free,} \\ \text{crossing-free}\}$, and $\mathcal{X}\in\{\text{line, segment}\}$, does $H$ have an $\mathcal{S}\ \mathcal{C}\ \mathcal{X}$ representation?
\end{problem}
In fact, this leads to eight problem variants for the combinations of $\mathcal{S},\mathcal{C},\mathcal{X}$, which we all discuss. We often omit \emph{non-crossing-free} and \emph{non-strict}, i.e., by a segment representation we mean a non-strict non-crossing-free segment representation.

For simplicity of our presentation, we disallow two hyperedges to contain the same set of vertices.
With this in mind, let us point out the following observation, which helps combine some of our complexity results.
\begin{observation}\label{observation:strictvnonstrict}
    Let $H=(V,E)$ be a hypergraph. If $H$ contains a pair of hyperedges $e,e'$ with $|e\cap e'|\ge 2$ and $e\ne e'$, then $H$ \revision{has no line representation, no strict line representation, no crossing-free line representation, and no strict crossing-free line representation.} Otherwise, $H$ is linear. Hence, $H$
    \begin{itemize}
        \item has a line representation if and only if it has a strict line representation, and
        \item has a crossing-free line representation if and only if it has a crossing-free strict line representation.
    \end{itemize}
\end{observation}

We investigate complexity properties of all these problems, and show several of them are $\exists \RR$-hard.
Here the class $\exists \RR$ is a complexity class between \NP\ and \PSPACE\ that contains all problems that can be reduced to solving an existentially quantified formula of polynomial equations and inequalities \revision{with integer coefficients}; this means that $\exists\RR$-hardness implies \NP-hardness.

\section{Pseudoline Stretchability and the Pappus Configuration} \label{section:pappus}
We often reduce from the $\exists\RR$-hard problem \textsc{pseudoline stretchability}~\citep{s-csgtp-09} to prove $\exists\RR$-hardness of our problems. A {\em pseudoline} is an $x$-monotone curve in $\mathbb{R}^2$ and a {\em simple pseudoline arrangement} is a set of pseudolines where every pair of pseudolines intersects exactly once and no three pseudolines meet at a common point. The \textsc{pseudoline stretchability} problem takes a simple pseudoline arrangement as input and seeks a combinatorially equivalent drawing where each pseudoline is drawn as a straight line segment, i.e., a homeomorphic line segment arrangement (see  \cref{fig:pappusgadget}(a)--(b)).

\begin{figure}[hb]
    \centering
    \includegraphics[width=\textwidth]{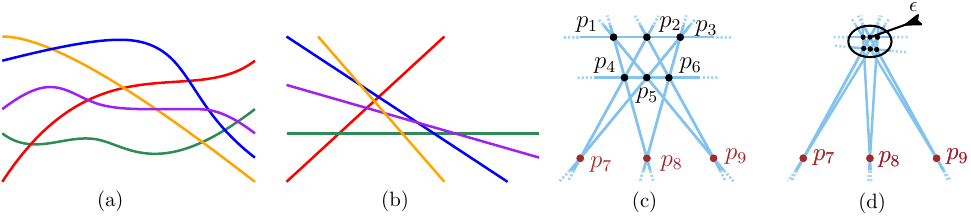}
    \caption{(a) A simple pseudoline arrangement and (b) its stretched line segment representation. (c) A line representation of the Pappus gadget. The three anchors $p_7,p_8,p_9$ are always collinear. (d) ``Pushing'' the points $p_1$ and $p_3$ very close to $p_2$ forces all fixers $p_1$--$p_6$ to be inside a disk of arbitrarily small radius $\epsilon$.}
    \label{fig:pappusgadget}
\end{figure}

Many of our reductions make use of the \emph{Pappus hypergraph} or \emph{Pappus configuration}, and its ability to force three points to be collinear, as stated below.

\begin{definition}
    The Pappus hypergraph or Pappus gadget $H_P=(V_P,E_P)$ is given as $V_P=\{p_1,\dots,p_9\}$ and $E_P$ consists of the following hyperedges.
    \[
        \begin{array}{l}
            \{p_1, p_2, p_3\}, \{p_4, p_5, p_6\}, \{p_1, p_4, p_8\}, \{p_1, p_5, p_9\}, \\
            \{p_2, p_4, p_7\}, \{p_2, p_6, p_9\}, \{p_3, p_5, p_7\}, \{p_3, p_6, p_8\}
        \end{array}
    \]

    We call $\{p_1,\dots,p_6\}$ the \emph{fixers} and $\{p_7,p_8,p_9\}$ the \emph{anchors}.
\end{definition}

For a line representation of the Pappus gadget see \cref{fig:pappusgadget}(c). Notice that all hyperedges have size three, all fixers have degree three, and all anchors have degree two.
The most important property is given next, and follows directly from the well-known Pappus theorem.
\begin{theorem}{\citep[Chapter 3.5]{coxeter1967geometry}}\label{thm:pappusthm}
    In every line or line segment representation (strict and non-strict) of the Pappus gadget the anchors are collinear.
\end{theorem}

{Intuitively, our reductions use newly introduced vertices as fixers to force existing vertices which are the anchors to be represented in a collinear way.
Next, we show that given a fixed embedding of the anchors, the remaining part of the Pappus gadget can be represented somewhat flexibly such that it can avoid a finite set of predefined points and lines.
\begin{restatable}{lemma}{lemmapappusfreedom}\label{lemma:pappusfreedom}
    Let $H_P$ be a Pappus gadget such that $\gamma_1,\gamma_2$, and $\gamma_3$ are three distinct collinear points, let $A$ be a finite set of points disjoint from $\{\gamma_1,\gamma_2,\gamma_3\}$, and let $B$ be a finite set of lines.
    There exists a line representation $(\alpha,\beta)$ of $H_P$ with $\alpha(p_7)=\gamma_1,\alpha(p_8)=\gamma_2,\alpha(p_9)=\gamma_3$, where none of the lines in $\beta(E_P)$ passes through any of the points in $A$, and none of the fixers lie on any line in $B$.
\end{restatable}
\begin{proof}
    Due to symmetries in the Pappus gadget we can assume that $\gamma_2$ lies on the line segment from $\gamma_1$ to $\gamma_3$. Now let $p$ be a point in the plane such that
    \begin{enumerate}
        \item $p$ is not on the line through $\gamma_1,\gamma_2,\gamma_3$,
        \item $p$ does not lie on any line in $B$, and
        \item the three lines passing through $p$ and $\gamma_1$, $p$ and $\gamma_2$, and $p$ and $\gamma_3$ do not contain any point from $A$.
    \end{enumerate}
    We now argue that for any $\epsilon>0$, $H_P$ can be represented such that $\alpha(p_i)=\gamma(p_i)$ for $i=7,8,9$, and the fixers of $H_P$ are represented inside the disk of radius $\epsilon$ around $p$ (see  \cref{fig:pappusgadget}(c)). Since $\gamma_2$ lies on the line segment from $\gamma_1$ to $\gamma_3$, we can represent $H_P$ such that the representation is homeomorphic to the one in \cref{fig:pappusgadget} (c) and $p_2$ lies exactly on $p$. This is because we can freely choose three collinear points defining the representation of $p_1,p_2$ and $p_3$ on a line different from the one through $p_7,p_8,p_9$. Together with the representation of $p_7,p_8,p_9$ this in turn defines the remaining representation of the Pappus gadget.
    By decreasing both the distance of $p_1$ to $p_2$ and $p_3$ to $p_2$ (i.e.\ pushing them towards $p_2$), all fixers of the Pappus gadget get closer to $\alpha(p_2)=x$ (see \cref{fig:pappusgadget}(d)). Thus, by making $\epsilon$ small enough, none of the following hyperedge representations contains any point of $A$, as they get ``closer and closer'' to the lines defined by $p$ and $\gamma_1,\gamma_2,\gamma_3$, respectively.
    \[\{p_1,p_4,p_8\},\{p_1,p_5,p_9\},\{p_2,p_4,p_7\},\{p_2,p_6,p_9\},\{p_3,p_5,p_7\},\{p_3,p_6,p_8\}\]
    We choose a small $\epsilon$ such that the disk around $p$ with radius $\epsilon$ neither contains any point from $A$ nor intersects with any line from $B$. Lastly, we show that the two remaining hyperedges $\{p_1,p_2,p_3\}$ and $\{p_4,p_5,p_6\}$ can be represented such they do not contain any point of $A$. To this end, we rotate the representation of $\{p_1,p_2,p_3\}$ around the point $p$ by some small angle $\epsilon'$ and adjust the representation accordingly. Notice that this also changes the slope of the representation of $\{p_4,p_5,p_6\}$.
\end{proof}

\section{Representations without bends}\label{section:representations}
\cref{table:resultswcrossings} shows our results for deciding whether a representation exists (irrespective of line or line segment crossings), which are proven in this section. The equivalence between results for line representations and strict line representations in \cref{table:resultswcrossings} comes from \cref{observation:strictvnonstrict}.

\subsection{Complexity results for lines}\label{section:representationslines}

\subparagraph*{Hardness results.}
We show that it is $\exists\RR$-hard to decide whether there exists a line representation for a given rank-$3$ hypergraph $H$.
We reduce from \textsc{Matroid Representability} \citep{DBLP:journals/corr/abs-2301-03221}. For the purposes of a simple description, we give here a simplified description of a variant of that problem that is still $\exists\RR$-hard \citep{DBLP:journals/corr/abs-2301-03221}.
We start with definitions.
A \emph{matroid} $M$ is given as $M=(X,\mathcal{I})$ where $X$ is the finite \emph{ground set} and $\mathcal{I}\subseteq 2^X$ is the set of independent sets with
\begin{enumerate}
    \item $\emptyset \in\mathcal{I}$,
    \item $I'\subset I\in \mathcal{I}$ implies $I'\in \mathcal{I}$, and
    \item $I_1,I_2\in \mathcal{I}$ with $|I_1|<|I_2|$ implies that there is an $x\in I_2\setminus I_1$ with $I_1\cup \{x\}\in\mathcal{I}$.
\end{enumerate}
A \emph{representation} of $M$ is an injective mapping $f(X):X\to \RR^3$ such that for any $Y\subseteq X$ we have $Y\in \mathcal{I}$ if and only if $f(Y)$ forms a set of linearly independent vectors in $\RR^3$.
The $\exists\RR$-hard problem \textsc{Matroid Representability} is given as input a matroid and the question is whether there is a representation $f$ of $M$.
For the vectors $v\in\RR^3$ we call the first, second, and third coordinate the $x,y,$ and $z$-coordinates, respectively.
We start by making some normalizations to $M$.
\begin{enumerate}
    \item First, we can assume that every independent set $I\in\mathcal{I}$ has cardinality at most $3$, as otherwise there is  no representation.
    \item Second, we can assume that each pair $\{x,x'\}\in {X\choose 2}$ forms an independent set, i.e.\ $\{x,x'\}\in\mathcal{I}$.
          Otherwise, $f(x)=cf(x')$ for some $c\in \RR$ must hold for any representation.
          We can remove $x'$ from $X$ and replace any occurrence of $x'$ in $\mathcal{I}$ by $x$, and obtain an equivalent instance w.r.t.\ representability.
\end{enumerate}

We are ready to prove the theorem.  %
\begin{restatable}{theorem}{erlinesrankthreethm}\label{thm:erlinesrank3}
    It is $\exists\RR$-hard to decide whether a rank-3 linear hypergraph has a line representation.
\end{restatable}
\begin{proof}
    We reduce from \textsc{Matroid Representability}. We are given a matroid $M=(X,\mathcal{I})$ (we assume both discussed normalizations were applied already) and transform it to a hypergraph $H=(V,E)$ as follows. To construct $H$, we first add to $H$ the set $X$ as vertices.
    Now consider a triple $t:=\{x,x',x''\}\in \binom{X}{3}$.
    \begin{itemize}
        \item If $t$ does not form an independent set in $M$, we introduce a new Pappus gadget and let $t$ be its anchor (see \cref{fig:pappusgadget}(c)).
        \item Otherwise, we add a new point $d$, force $d$ to be collinear with $x$ and $x'$ using a new Pappus gadget with anchors $d,x$ and $x'$, and lastly add the hyperedge $\{d,x''\}$ to $E$ (see \cref{fig:erstrictlinegadets}(a)).
    \end{itemize}
    We argue that $H$ has a line representation if and only if $M$ has a representation:

    ``$\Rightarrow$'': let $(\alpha,\beta)$ be a line representation of $H$.
    For $x\in X$, let $\alpha(x)=(r,s)^T$. We set $f(x)=(r,s,1)^T$ and claim that $f$ is a representation of $M$.
    Consider any triple $t:=\{x,x',x''\}\in \binom{X}{3}$.
    \begin{itemize}
        \item If $t$ forms an independent set then $x,x',x''$ are not collinear: if they were, then consider the line representation of the corresponding gadget (\cref{fig:erstrictlinegadets}(a)).
              As $x,x',x''$ would be collinear, the line through $d$ and $x''$ would also pass through $x'$ -- a contradiction, because we assumed that $(\alpha,\beta)$ is a line representation.
        \item Otherwise, the Pappus gadget involving $t$ forces $x,x',x''$ to be on a line in $(\alpha,\beta)$ by \cref{thm:pappusthm} -- and hence $x,x',x''$ are on a common hyperplane containing the origin in $f)$.
    \end{itemize}

    ``$\Leftarrow$'': let $f$ be a representation of $M$.
    First, if any $f(x')$, $x'\in X$, has $z$-coordinate $0$ we multiply all $f(x)$, $x\in X$, by the same rotation matrix $R\in SO(3)\subseteq\RR^{3\times 3}$ such that no $f(x)$ has $z$-coordinate $0$. This is possible, as $X$ is finite.
    Second, we scale each $f(x)$ by $1$ divided by its $z$-coordinate so that the $z$-coordinates of all vectors in $f(X)$ are $1$. The new $f$ is still a representation.
    For $x\in X$, we set $\alpha(x)$ equal to the point defined by the first two coordinates of $f(x)$. Essentially, we applied a projective transformation. Lastly, we have to find a valid representation for the remaining vertices of $H$.
    Consider the triples $t:=\{x,x',x''\}\in \binom{X}{3}$  one by one. We find a representation for its involved gadgets one by one.
    \begin{itemize}
        \item If $t$ forms a dependent set in $M$, we place the Pappus gadget containing $t$ such that none of the lines of the gadget introduce unwanted incidences. This is always possible as $\alpha(x),\alpha(x'),$ and $\alpha(x'')$ are collinear and because of \cref{lemma:pappusfreedom}.
        \item If $t$ forms an independent set in $M$, let $d$ be the newly introduced gadgets point. We first set $\alpha(d)$ to be a point on the line segment connecting $\alpha(x)$ and $\alpha(x')$, such that the line $\beta(\{d,x''\})$ through $d$ and $x''$ has no unwanted incidences. This is possible as there is certainly no line through $x$ and $x'$ by construction. Lastly, we represent the Pappus gadget containing $x,x'$, and $d$, such that there are no unwanted incidences. Again, this is always possible because of \cref{lemma:pappusfreedom}.\qedhere
    \end{itemize}
\end{proof}

\begin{figure}[t]
    \centering
    \includegraphics{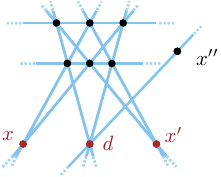}
    \caption{A gadget forcing $x,x',$ and $x''$ not to be collinear in the proof of \cref{thm:erlinesrank3}.}
    \label{fig:erstrictlinegadets}
\end{figure}

The reduction in \cref{thm:erlinesrank3} produces a linear hypergraph without two degree-1 vertices being part of the same hyperedge. Hence, applying a standard point-line duality transformation to a line representation $H$ results in a representation of its max-degree-$3$ dual hypergraph and vice versa. This is enough to show that $\exists\RR$-hardness also holds for max-degree-$3$ hypergraphs.
\begin{corollary}\label{thm:erlinesdeg3}
    It is $\exists\RR$-hard to decide whether a max-degree-3 hypergraph has a line representation.
\end{corollary}

\subparagraph*{Solvable cases.}
By definition, non-linear hypergraphs do not have line representations. However, every rank-2 hypergraph (which is linear) and every linear max-degree-2 hypergraph has a strict line representation by starting with points or lines in general position, respectively.
\begin{observation}\label{obs:linerank-2}
    Every rank-2 hypergraph and every linear max-degree-2 hypergraph has a line representation.
\end{observation}
\subsection{Complexity results for line segments}\label{section:representationslinesegments}

\subparagraph*{Hardness results.}
The first result is that deciding whether a hypergraph has a segment representation is $\exists\RR$-hard, even for hypergraphs of constant rank and constant maximum degree.
\begin{theorem}\label{thm:erlinesgmentsrank3}
    It is $\exists\RR$-hard to decide whether a rank-$3$ max-degree-$6$ hypergraph has a segment representation.
\end{theorem}
\begin{proof}
    We reduce from \textsc{pseudoline stretchability}~\citep{s-csgtp-09}.
    Given an instance $I$ of \textsc{pseudoline stretchability},
    we construct a hypergraph~$H$ by taking intersection points as vertices. For each pseudoline with the intersection points $v_1, v_2,\dots,v_t$ in this order, we construct hyperpaths consisting of the hyperedges $\{v_1,v_2,v_3\},\{v_2,v_3,v_4\}, \dots,\{v_{t-2}, v_{t-1}, v_t\}$ (\cref{fig:nonstrHard}(a)--(b)).

    By the construction of $H$, a solution to the \textsc{pseudoline stretchability} instance corresponds to a representation with straight line segments (\cref{fig:nonstrHard}(c)): We assume that $H$ has a representation with straight line segments. Observe that then the union of hyperedges forming the hyperpaths as above appear together as one single line segment with the vertices in the correct order.
    Conversely, if $I$ is stretchable, then $H$ has a segment representation.
\end{proof}

\begin{figure}[h]
    \centering    \includegraphics[width=\linewidth]{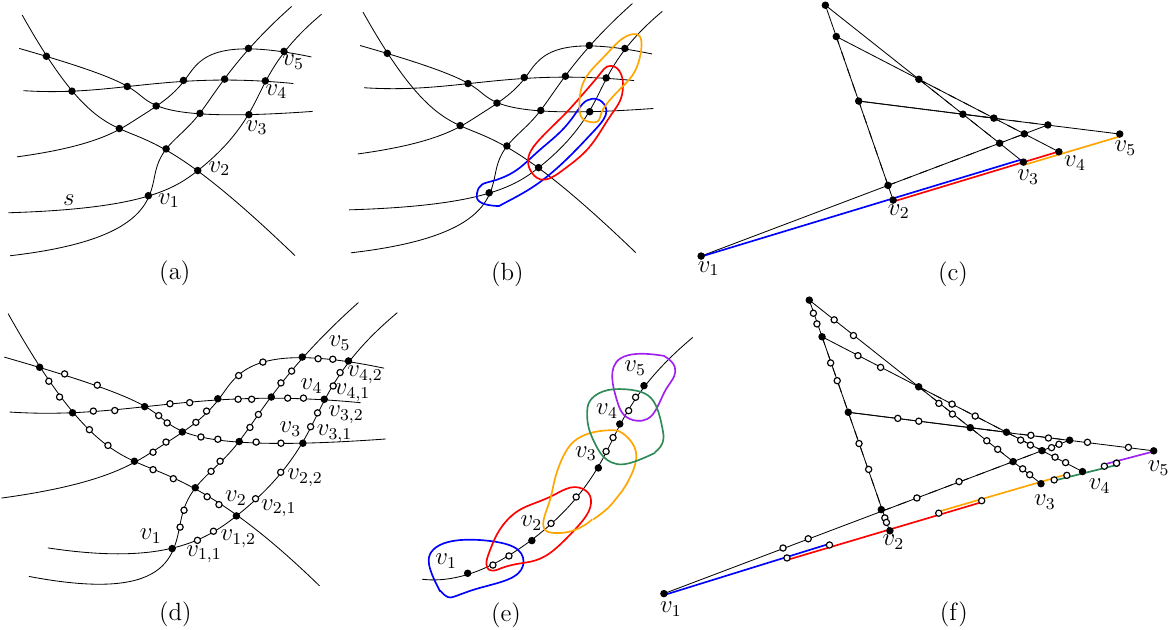}
    \caption{(a) A pseudoline arrangement. (b) Illustration of the hyperedges corresponding to a pseudoline as in the proof of \cref{thm:erlinesgmentsrank3}.  (c) A corresponding segment representation. (d)--(f) Illustration for the reduction of \cref{thm:erlinesgmentsmdeg2}. }
    \label{fig:nonstrHard}
\end{figure}

We show later that the problem becomes tractable for rank $2$ instead of rank~$3$. Surprisingly, this is not the case for maximum degree $2$ even for constant rank, as shown with a similar but more involved reduction.

\begin{theorem}\label{thm:erlinesgmentsmdeg2}
    It is $\exists\RR$-hard to decide whether a rank-$5$ max-degree-$2$ hypergraph has a segment representation.%
\end{theorem}
\begin{proof}
    We reduce from \textsc{pseudoline stretchability}~\citep{s-csgtp-09}. Let $I$ be an instance of \textsc{pseudoline stretchability} with at least five lines. To construct a hypergraph $H$, we first add to $H$ all the intersection points of $I$ as vertices.
    For each pseudoline with the intersection points $v_1, v_2,v_3,v_4,\dots,v_{t-2}, $ $v_{t-1},v_t$ in this order, we introduce the further vertices $v_{i,1},v_{i,2}$ for all $1\le i\le t-1$. We then add the hyperedges $\{v_1,v_{1,1},v_{1,2}\}$, $\{v_t,v_{t-1,1},v_{t-1,2}\}$, and $\{v_i,v_{i-1,1},v_{i-1,2},v_{i,1},v_{i,2}\}$ for all $2\le i\le t-1$. By the construction of $H$, in any representation $v_1,v_2,\dots,v_t$ must be collinear and appear in the correct order (\cref{fig:nonstrHard}(d)--(f)).
\end{proof}
\begin{figure}[ht]
    \centering    \includegraphics[width=.6\linewidth]{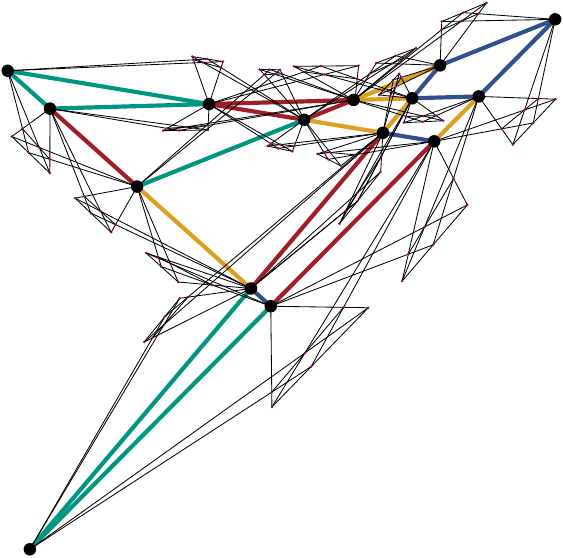}
    \caption{A strict segment representation of the hypergraph resulting from the reduction in \cref{thm:erstrictlinesegments} applied to the pseudoline arrangement in \cref{fig:nonstrHard}(a).}
    \label{fig:erstrictlinesegments}
\end{figure}

By careful use of the Pappus gadget, we can prove $\exists\RR$-hardness for strict segment representations, even when the hypergraphs are of rank $3$ and maximum degree $12$.
\begin{theorem}\label{thm:erstrictlinesegments}
    It is $\exists\RR$-hard to decide whether a rank-$3$ max-degree-$16$ hypergraph has a strict segment representation.
\end{theorem}
\begin{proof}
    Consider an instance of \textsc{pseudoline stretchability}. We build a hypergraph $H$ with intersection points as vertices. For each pseudoline let its intersection points be $v_1, v_2,\dots,v_t$. We need to force (1) all these points collinear, and (2) fix the correct order.
    For (1), we add for each triple $v_i,v_{i+1},v_{i+2}$ ($i=1, \dots, t-2$) a new Pappus gadget with the triple as its anchors.
    For (2), we  add the hyperedges $\{v_1,v_2\},\{v_2,v_3\},\dots,\{v_{t-1},v_t\}$.
    We claim that $I$ is stretchable if and only if $H$ has a strict segment representation, see \cref{fig:erstrictlinesegments}.

    ``$\Rightarrow$'': Consider the placement of intersection points as they appear in the stretched pseudoline arrangement. Every size-two hyperedge can be represented by drawing a line segment between their corresponding vertices. We also need to represent the Pappus gadgets, which is possible because of the freedom of the Pappus gadget (see \cref{lemma:pappusfreedom}).

    ``$\Leftarrow$'' We built $H$ such that in a strict segment representation the intersection points of a polyline are collinear and ordered correctly, making $I$ stretchable.
\end{proof}

\subparagraph*{Solvable cases.}
The first result is a direct consequence of \cref{obs:linerank-2} and by the fact that rank-2 hypergraphs are linear by our assumptions.
\begin{corollary}\label{corr:strictlinesegmentpoly}
    Every rank-2 hypergraph and every linear max-degree-2 hypergraph has a strict and non-strict segment representation.
\end{corollary}

As we have seen, testing whether a max-degree-2 hypergraph has a segment representation is $\exists\RR$-hard (\cref{thm:erlinesgmentsmdeg2}). However, we can at least characterize all rank-3 max-degree-2 hypergraphs that have a segment representation by a set of forbidden subhypergraphs:
\begin{itemize}\sloppy
    \item The \emph{rigid triangle} consists of vertices $a,b,c,d$ and hyperedges $e_1=\{a,b,c\}$ $e_2=\{b,c,d\}$, and $e_3=\{a,d\}$. The hyperedge $e_3$ can furthermore also contain another vertex $f$, i.e., $e_3=\{a,d,f\}$. See \cref{fig:fp}(a)--(b).
    \item The \emph{rigid parallel 2-path} consists of the vertices $a,b_1,b_2,c_1,c_2,d$, and the hyperedges $e_1=\{a,b_1,c_1\}$, $e'_1=\{b_1,c_1,d\}$, $e_2=\{a,b_2,c_2\}$, $e'_2=\{b_2,c_2,d\}$. See \cref{fig:fp}(c).
\end{itemize}

\begin{figure}[ht]
    \centering    \includegraphics[width=\textwidth]{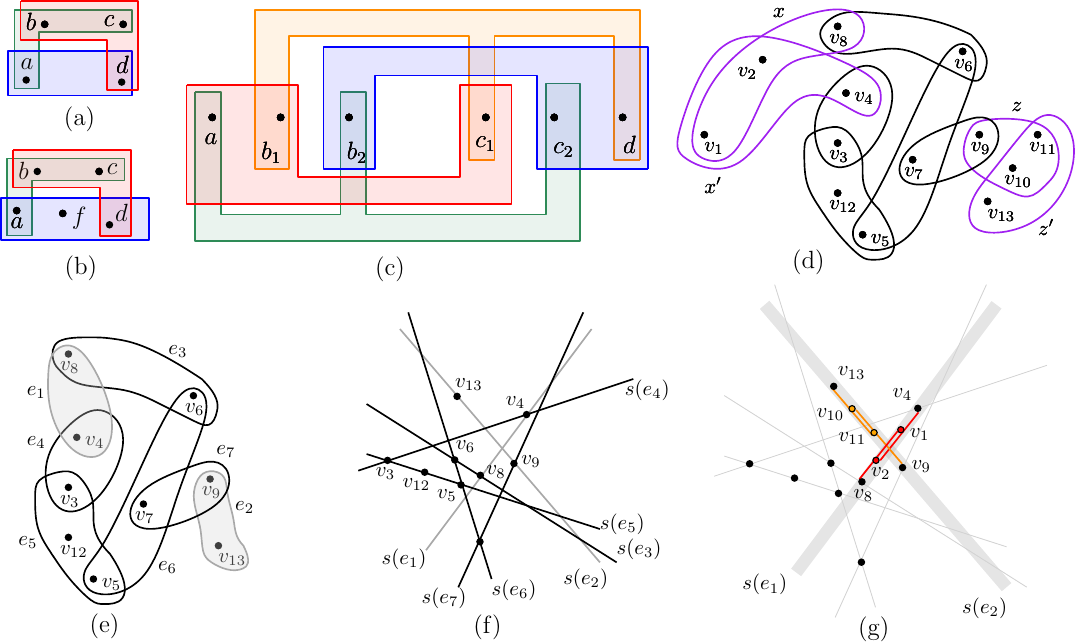}
    \caption{(a)--(b) Illustration for rigid triangles. (c) A rigid parallel 2-path. (d) Illustration for $E_b$ in purple where the sibling pairs are $x,x'$ and $z,z'$. (e) Illustration for $E_c$, where the hyperedges $((x\cup x')\setminus (x\cap x'))$ and $((z\cup z')\setminus (z\cap z'))$ corresponding to sibling pairs are shown in gray. (f) A segment representation for $E_c$, where sibling pairs $x,x'$ and $z,z'$ are mapped to $s(e_1)$ and $s(e_2)$, respectively. (g) Splitting $s(e_1)$ and $s(e_2)$ to create $\beta(x),\beta(x')$ and $\beta(z),\beta(z')$, respectively.}
    \label{fig:fp}
\end{figure}
\begin{restatable}{theorem}{polylinesegmentrankThreedegreeTwo}\label{thm:polylinesegmentrank3degree2}
    A rank-3 max-degree-2 hypergraph has a segment representation if and only if it does not contain a set of hyperedges that together form a rigid triangle or rigid parallel 2-path.
\end{restatable}
\begin{proof}
    The forward direction is clear as rigid triangles and rigid parallel 2-paths themselves do not have segment representation. For a rigid triangle (\cref{fig:fp}(a)--(b)), a segment representation of the hyperedges $e_1$ and $e_2$ forces  $a,b,c,d$ to be collinear where $b$ and $c$ must appear between $a$ and $d$. Hence, $e_3$ cannot be represented without an unnecessary adjacency.
    For a rigid parallel 2-path (\cref{fig:fp}(c)), a segment representation of the hyperedges $e_i$ and $e'_i$, where $1\le i\le 2$, forces  $a,b_i,c_i,d$ to be collinear where $b_i$ and $c_i$ must appear between $a$ and $d$.
    Since a segment representation of the hyperedges $e_1$ and $e_2$ forces $a$ to lie between $b_1,c_1$ and $b_2,c_2$, we cannot find a placement for $d$ to satisfy the constraints imposed by $e_i$ and $e'_i$.

    Assume now that we are given a hypergraph $H=(V,E)$ without any of the two forbidden substructures. We show how to construct a segment representation.
    Let $E_a\subseteq E$ be defined as
    \[E_a=\{e\in E\mid \forall e'\in E:e'\ne e\Rightarrow |e\cap e'|\le 1\},\]
    and let $E_b=E\setminus E_a$. In other words, a hyperedge belongs to $E_a$ if it has at most one vertex in common with any other edge in the hypergraph, and $E_b$ consists of the remaining set of hyperedges.

    Since $H$ is a max-degree-2 and rank-3 hypergraph, we can observe some structural properties of $E_b$. For each $e\in E_b$ there exists exactly one other $e'\in E_b$ with $|e\cap e'|=2$. Hence, we define $\text{sib}(e)=e'$ and partition $E_b$ into pairs of siblings (\cref{fig:fp}(d)).
    Now define the edge set $E_c$ which consists of $E_a$ plus for each sibling pair $\{x,x'\}$ it contains the hyperedge $p_{x,x'}=(x\cup x')\setminus (x\cap x')$ (\cref{fig:fp}(e)).
    Next, consider a line arrangement of $|E_c|$ line segments in general position such that each pair of line segments crosses exactly once.
    Assign to each hyperedge $e\in E_c$ a distinct  line segment $s(e)$ from this arrangement. If $e=p_{x,x'}$ for a sibling pair $x,x'$, we set $s(x)=s(x')=s(e)$.
    For $e\in E_a$ we set $\beta(e)=s(e)$. Next, for each vertex $v\in V$ that is not incident to two hyperedges of a sibling pair, we let $\alpha(v)$ be the intersection point of the line segments $s(e)$ and $s(e')$ of its two incident hyperedges $e$ and $e'$ (\cref{fig:fp}(f)).

    We now find the representations for the sibling pairs. It suffices to consider a sibling pair $p_{x,x'}$ such that $|x\cup x'|=4$, otherwise the construction is trivial.
    For each sibling pair $x,x'$ we cover $s(p_{x,x'})$ by $\beta(x)$ and $\beta(x')$ as follows.  Firstly, $\beta(x)$ and $\beta(x')$ need to share a non-empty segment that will contain the two points $\alpha(y)$ and $\alpha(y')$ for $y,y'\in x\cap x'$ (\cref{fig:fp}(g)). Since the line segments are in general position, $\alpha(y)$ and $\alpha(y')$ can be chosen such that no other line segment besides $\beta(p_{x,x'})$ contains them. Next, if for example $x\setminus x'=\{v\}$, then we define $\beta(x)$ such that it contains $\alpha(v)$ and $\beta(x')$ such that it does not contain $\alpha(v)$. If $x'\setminus x=\{v'\}$ we define $\beta(x')$ such that it contains $\alpha(v')$ and $\beta(x)$ such that it does not contain $\alpha(v')$. Notice that all the requirements above are possible by splitting $s(p_{x,x'})$ into three consecutive parts, the first being covered by $\beta(x)$ and containing $\alpha(v)$, the second being covered by both $\beta(x)$ and $\beta(x')$ and containing $\alpha(y)$ and $\alpha(y')$, and the third being covered by $\beta(x')$ and containing $\alpha(v')$.
    Suppose for a contradiction that the above construction generates unnecessary adjacencies among the four vertices on $\beta(p_{x,x'})$. The sibling pair enforces $y,y'$ to lie between $v,v'$. Since $v$ and $v'$ each are of degree 2, an unnecessary adjacency can appear only if there exists another set $S$ of hyperedges that force a set of vertices to be collinear with $v$ and $v'$. If $|S|=1$, then we obtain a rigid triangle. If $|S|>1$, then we obtain a rigid parallel 2-path. Since $H$ does not contain these structures, we obtain a correct representation.
\end{proof}

Notice that testing whether $H$ contains a rigid triangle or a rigid parallel 2-path can be done in polynomial time, so we have identified a problem variant that is polynomial-time solvable.
We do not know the complexity for hypergraphs with maximum degree 2 and rank 4.

\section{Representations without bends and without crossings}\label{section:crossfreerepresentations}
In this section, we give some results on deciding the existence of crossing-free representations, which are summarized in \cref{table:resultswocrossings}. Again, the equivalence between strict and non-strict is because of \cref{observation:strictvnonstrict}.

\subsection{Complexity results for lines}\label{section:crfreerepresentationslines}
We have two positive results for rank-2 and max-degree-2 hypergraphs, respectively. The first follows due to an easy case distinction, the second by an equivalence of crossing-free line representability of a linear hypergraph $H$ (non-linear hypergraphs are not line representable) and its hyperedge intersection graph being a complete $k$-partite graph.
\begin{restatable}{theorem}{crossfreelines}\label{thm:crossfreelines}
    We can decide in polynomial time whether a rank-$2$ hypergraph has a crossing-free line representation.
\end{restatable}
\begin{proof}
    Let $H$ be a hypergraph. We assume $H$ is linear, otherwise there is no such representation.
    We consider different cases.
    \begin{itemize}
        \item If $H$ only contains vertices of degree one then $H$ has a crossing-free line representation with parallel lines as it consists of pairwise independent hyperedges.
        \item Otherwise, let $v$ be a vertex of degree $d>1$ and let $e_1,\dots,e_d$ be its incident hyperedges. We again consider two cases.
              \begin{itemize}
                  \item If $d=2$ then $H$ has a crossing-free line representation if and only if the set of  hyperedges $E' = E\setminus\{e_1,e_2\}$ can be represented as a set of parallel lines without creating unnecessary adjacencies with the representations of $e_1$ and $e_2$. This implies that no two hyperedges in $E'$ share a common vertex, and exactly one of the following conditions hold: (a) Every hyperedge of $E'$ has a vertex in common with $e_1$. (b) Every hyperedge of $E'$ has a vertex in common with $e_2$. (c) Every hyperedge of $E'$ has a vertex in common with $e_1$ and a vertex in common with $e_2$.

                        For (a) and (b), the hyperedges of $E'$ can be drawn parallel to $e_1$ and $e_2$, respectively. For (c), the hyperedges of $E'$ can be drawn with parallel lines where all of them intersect the lines for $e_1$ and $e_2$.
                  \item If $d\ge 3$ then $H$ cannot contain any hyperedge $e\not\in\{e_1,\dots,e_d\}$, because the representation of $e$ would intersect all the other lines implying a rank of at least 3. $H$ can be represented using a set of lines intersecting at a common point $\alpha(v)$.
              \end{itemize}
    \end{itemize}
    Each of the above cases can be processed by iterating over the hyperedges a polynomial number of times. Consequently, the overall running time is polynomial.
\end{proof}

\begin{restatable}{theorem}{crossfreelinesdegtwo}\label{thm:crossfreelinesdeg2}
    Let $H=(V,E)$ be a max-degree-2 hypergraph. Then, in polynomial time, one can decide whether $H$ admits a crossing-free line representation.
\end{restatable}
\begin{proof}
    If $H$ is not linear, then it does not have a crossing-free line representation. Thus, assume that $H$ is linear.
    Any two hyperedges that share a vertex must intersect. Since the representation is crossing free, every intersection point must correspond to a vertex.
    Therefore, we construct the graph $G$ having as vertex set $E$ and distinct $e,e'\in E$ are adjacent if and only if they share a vertex. 
    
    In a crossing-free  line representation, a maximal set of parallel lines must intersect all the other remaining lines, and every intersection determines a vertex. Therefore, if $H$ admits a crossing-free line representation, then $G$ is complete $k$-partite for some $k\in\mathbb{N}$.
    
   Assume now that $G$ is complete $k$-partite for some $k\in\mathbb{N}$, which can be tested in polynomial time by checking whether each connected component in the complement graph of $G$ is a clique. We can represent $H$ by choosing $k$ representative lines in general position,  and then creating a set of parallel lines for each of these representative lines whose cardinality matches the size of the corresponding partite set in $G$.
\end{proof}

\subsection{Line segments}\label{section:crfreerepresentationslinesegments}
For line segments, we first present hardness results and then cases that we solve, in particular this will include classes of hypergraphs that always have a crossing-free segment representation.
\subparagraph*{Hardness results.}
The reduction in the proof of \cref{thm:erlinesgmentsrank3} constructs a hypergraph that admits a crossing-free segment representation if and only if the corresponding pseudoline arrangement is stretchable. Thus, the same reduction gives the following result.
\begin{corollary}\label{cor:ercrossingfreelinesegmentsrank3}
    It is $\exists\RR$-hard to decide whether a rank-$3$ max-degree-$6$ hypergraph has a crossing-free segment representation.
\end{corollary}
The same holds for the reduction in the proof of \cref{thm:erlinesgmentsmdeg2}, which gives the following.
\begin{corollary}\label{cor:ercrossingfreelinesegmentsmdeg2}
    It is $\exists\RR$-hard to decide whether a rank-$5$ max-degree-$2$ hypergraph has a crossing-free segment representation.
\end{corollary}
\begin{figure}[ht]
    \centering    \includegraphics[width=\linewidth]{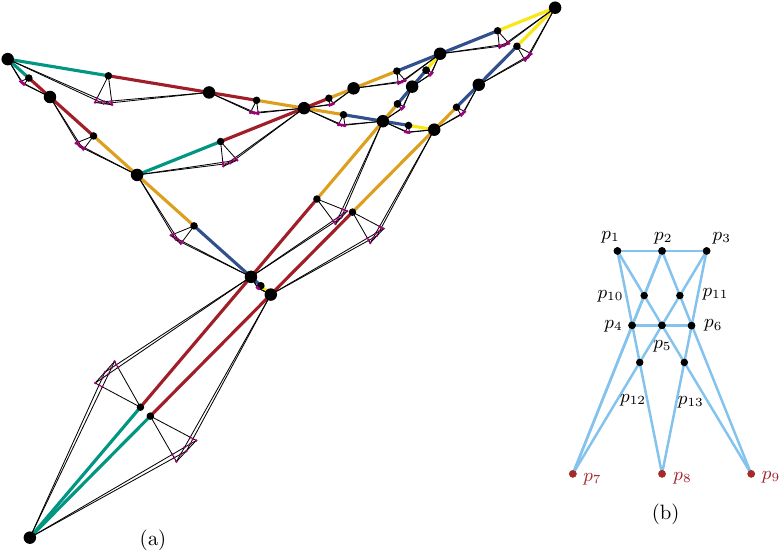}
    \caption{(a) A strict crossing-free segment representation of the hypergraph resulting from the reduction in \cref{thm:erstrictlinesegments} applied to the pseudoline arrangement in \cref{fig:nonstrHard}(a). (b) A strict crossing-free segment representation of the extended Pappus gadget.}
    \label{fig:ercrossingfreestrictlinesegments}
\end{figure}
With a slightly more involved reduction from stretchability using an \emph{extended Pappus gadget} (intersection points are replaced by vertices), we can show that it is $\exists\RR$-hard to decide whether a rank-$5$ max-degree-$10$ hypergraph has a strict crossing-free segment representation.
\begin{restatable}{theorem}{ercrossingfreestrictlinesegments}\label{thm:ercrossingfreestrictlinesegment}
    It is $\exists\RR$-hard to decide whether a rank-$5$ max-degree-$10$ hypergraph has a strict crossing-free segment representation.
\end{restatable}
\begin{proof}
    Let $I$ be an instance of \textsc{pseudoline stretchability} with at least $3$ lines. We construct a hypergraph $H$ by taking intersection points as vertices. For each pseudoline let its intersection points be $v_1, v_2,v_3,v_4,\dots,v_{t-2},v_{t-1},v_t$. We need to construct gadgets such that
    \begin{itemize}
        \item[] (1) all these points are collinear, and
        \item[] (2) they appear in the correct order.
    \end{itemize}
    For (1), consider two consecutive intersection points $v_{i},v_{i+1}$. We add a new vertex $w_i$. Now we want to force $v_i,w_i$ and $v_{i+1}$ to be collinear. We cannot use the Pappus gadget as the Pappus gadget itself has crossings. However, we can extend the Pappus gadget with new vertices that correspond to these intersection points (see \cref{fig:ercrossingfreestrictlinesegments}(b)).
    We call the hypergraph corresponding to the representation in \cref{fig:ercrossingfreestrictlinesegments} the \emph{extended Pappus gadget}. As the Pappus gadget is a subhypergraph of the extended Pappus gadget, in any strict crossing-free segment representation the \emph{anchors} $p_7,p_8,p_9$ must be collinear.
    Thus, for forcing the triple $v_i,w_i,v_{i+1}$ to be collinear, we introduce a new extended Pappus gadget and let the triple be its anchors. It is important that $v_i$ corresponds to $p_7$, $w_i$ to $p_8$, and $v_{i+1}$ to $p_9$, as we do not have the symmetries that were present in the standard Pappus gadget anymore.
    Next, for each triple $v_i,v_{i+1},v_{i+2}$ of consecutive intersection points we add the hyperedge $w_i,v_{i+1},w_{i+1}$. Together, the extended Pappus gadgets and these line segments force $v_1,v_2,\dots,v_t$ to be collinear.
    For (2), the above construction already ensures that $v_i,v_{i+1},v_{i+2}$ appear in the correct order for $2\le i\le t-3$. We simply add two more hyperedges $\{v_1,w_1\}$ and $\{w_{t-1},v_t\}$ to ensure the order for all $1\le i\le t-2$.

    We claim that $I$ is stretchable if and only if $H$ has a crossing-free strict segment representation (this correspondence is shown in \cref{fig:ercrossingfreestrictlinesegments}(a)).

    ``$\Rightarrow$'': Consider the representation of intersection points as they appear in the stretched pseudoline arrangement and for $v_i,v_{i+1}$ belonging to some pseudoline, represent $w_i$ somewhere along the line segment connecting $v_i$ and $v_{i+1}$. The hyperedges corresponding to the hyperedges of size two and three not belonging to the extended Pappus gadgets can be represented and pass through their two/three incident vertices. We also need to represent the extended Pappus gadgets. Due to arguments similar to \cref{lemma:pappusfreedom}, the extended Pappus gadget points that are not the anchors can be placed into an infinitesimally small disk that is arbitrarily centered. Hence, we can place the Pappus gadgets such that the line segments of the gadgets have no unwanted incidences of points (vertices) and lines (hyperedges), nor unwanted crossings. Hence, $H$ can be represented strictly by line segments.

    ``$\Leftarrow$'' We constructed $H$ such that in any strict crossing-free segment representation the intersection points of a polyline appear collinear and in the correct order. Hence, $I$ is stretchable.
\end{proof}

A correspondence between strict crossing-free segment representations of degree-$2$ hypergraphs and segment intersection graphs gives the following result.

\begin{restatable}{theorem}{erstrictlinesegmentmdegtwo}\label{thm:erstrictlinesegmentmdeg2}
    It is $\exists\RR$-hard to decide whether a max-degree-$2$ hypergraph has a strict crossing-free segment representation.
\end{restatable}
\begin{proof}
    We reduce from the problem \textsc{segment intersection graph recognition}, which is $\exists\RR$-hard \citep{DBLP:journals/corr/Matousek14}. This problem is given a simple graph $G$ and asks whether there exists an arrangement $\mathcal{A}$ of segments in the plane such that $G$ is the intersection graph of $\mathcal{A}$, i.e., the vertices of $G$ correspond to segments of $\mathcal{A}$, and two vertices are connected by an edge if and only if the corresponding segments share at least one point. In that case, we say that $\mathcal{A}$ \emph{represents} $G$. We can assume that $G$ is connected and has at least two vertices as this restriction preserves the $\exists\RR$-hardness result of \citet{DBLP:journals/corr/Matousek14}.
    We construct a hypergraph $H$ as follows. For each edge $e$ in $G$ we add a vertex $v_e$ to $H$. For each vertex $v$ in $G$ that is incident to edges $e_1,e_2,\dots,e_\ell$ we add the hyperedge $\{v_{e_1},v_{e_2},\dots,v_{e_\ell}\}$ to $H$. Notice that $H$ has degree two as each edge in $G$ is incident to two vertices.
    We show that $G$ is a segment intersection graph if and only if $H$ has a strict crossing-free segment representation.

    ``$\Rightarrow$'': Let $\mathcal{A}$ be a representation of $G$. As $G$ is simple, we can assume that no two segments in $\mathcal{A}$ lie on the same line (see also \citep{DBLP:journals/corr/Matousek14}). Thus, if two line segments share a non-empty intersection, this intersection consists of a single point. Furthermore, with standard perturbation arguments we can assume that no three line segments in $\mathcal{A}$ intersect at a single point, see e.g.~\citep[Section 5]{eppstein2009testing}.
    It is now easy to turn $\mathcal{A}$ into a strict crossing-free segment representation $(\alpha,\beta)$. For vertices $v_e$ in $H$, $\alpha(v_e)$ is simply the unique intersection point of the two line segments corresponding to the two incident vertices of $e$. For the hyperedge $\hat{e}=\{v_{e_1},v_{e_2},\dots,v_{e_\ell}\}$, $\beta(\hat{e})$ is the line segment in the line arrangement passing through $\alpha(v_{e_1}),\alpha(v_{e_2}),\dots,\alpha(v_{e_\ell})$.

    ``$\Leftarrow$'': Let $(\alpha, \beta)$ be a strict crossing-free segment representation of $H$. Notice that when interpreting the representation as line segment arrangement $\mathcal{A}$, then $\mathcal{A}$ represents $G$.
\end{proof}

\paragraph*{Solvable cases.}
We now present some cases where the existence of a crossing-free segment representation can be decided in polynomial time.

The rank-2 case corresponds to deciding whether a graph is planar.
\begin{observation}\label{obs:segmentrank2planarity}
    Let $H$ be a rank-$2$ hypergraph without hyperedges of size $1$. Then $H$ is a simple graph, and every crossing-free segment representation can be made strict. Furthermore, $H$ has such a representation if and only if it is planar.
\end{observation}
As hyperedges of size $1$ do not affect segment representability, this solves the problem for rank-2 hypergraphs.

We now present some hypergraph classes for which we can decide the existence of a strict crossing-free segment representation in polynomial time. The first class is based on \revision{\emph{permutation graphs} which are graphs with vertex set $\{v_1,\dots,v_n\}$ and there exists a permutation $\pi$ such that $v_i,v_j$, $i<j$, are adjacent if and only if $i$ comes after $j$ in $\pi$.}

\begin{restatable}{theorem}{permgraphthm}\label{thm:permgraph}
    Let $H$ be a degree-2 linear  hypergraph and let $G$ be its hyperedge intersection graph. If $G$ is a permutation graph (which is verifiable in linear time), then $H$ admits a crossing-free strict segment representation. %
\end{restatable}
\begin{proof}
    If $G$ is a permutation graph, then by definition $G$ admits a representation $\Gamma$ with segments such that two segments intersect if and only if the corresponding vertices are adjacent in $G$. Furthermore, the segments lie between a pair of parallel lines $\ell_u$  and $\ell_b$, where each segment has one endpoint on $\ell_u$  and the other endpoint on $\ell_b$. One can determine in linear time whether a graph is a permutation graph, and, if so, then its segment representation can also be constructed in linear time~\citep{mcconnell1999modular}. Figures~\ref{fig:t20}(a)--(b) illustrate $H$ and the corresponding segment representation $\Gamma$ of its hyperedge intersection graph. Note that a perturbation of the endpoints on $\ell_u$ does not change the intersection representation as long as the order of the endpoints on $\ell_u$ is preserved, but it can ensure that
    no three  segments meet at a single point (\cref{fig:t20}(c)). %

\begin{figure}[h]
    \centering    \includegraphics[width=\linewidth]{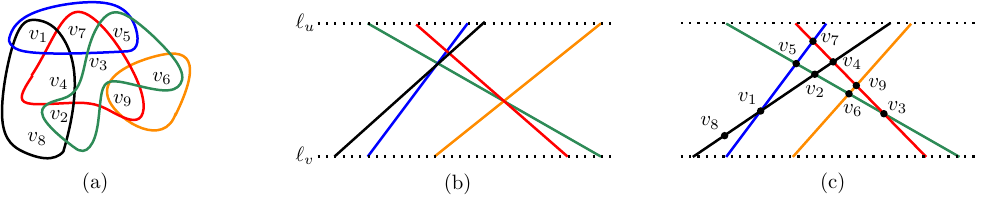}
    \caption{(a) A degree-2 linear hypergraph $H$.  (b) An intersection representation of the hyperedge intersection graph.   (c) A perturbation of the endpoints on $\ell_u$ and the representation of the vertices. }
    \label{fig:t20}
\end{figure}

Consider now the segments of $\Gamma$ as the hyperedges of $H$. Every vertex $v$  with degree two in $H$ can be represented at the intersection point of the hyperedges in $\Gamma$ that contain $v$. Finally, the vertices that are unique to a hyperedge can be placed by choosing distinct points on the segment representing that hyperedge.
\end{proof}

The \emph{vertex-edge incidence graph} of a hypergraph $H=(V,E)$ is a bipartite graph $G=(V \cup E, E')$, where  an edge $(v,e)\in E'$, $v\in V$ and $e\in E$, exists, if and only if $e$ contains $v$. We obtain the following result.

\begin{restatable}{theorem}{onesided}\label{lem:onesided}
    Let $H=(V,E)$ be a hypergraph with the  vertex-edge incidence graph $G=(V\cup E,E')$. If $G$ has a planar embedding with vertex set $E$ on the outerface, then we can decide whether $H$ admits a crossing-free strict segment representation in polynomial time.
\end{restatable}
\begin{figure}[pt]
    \centering
    \includegraphics[width=\textwidth]{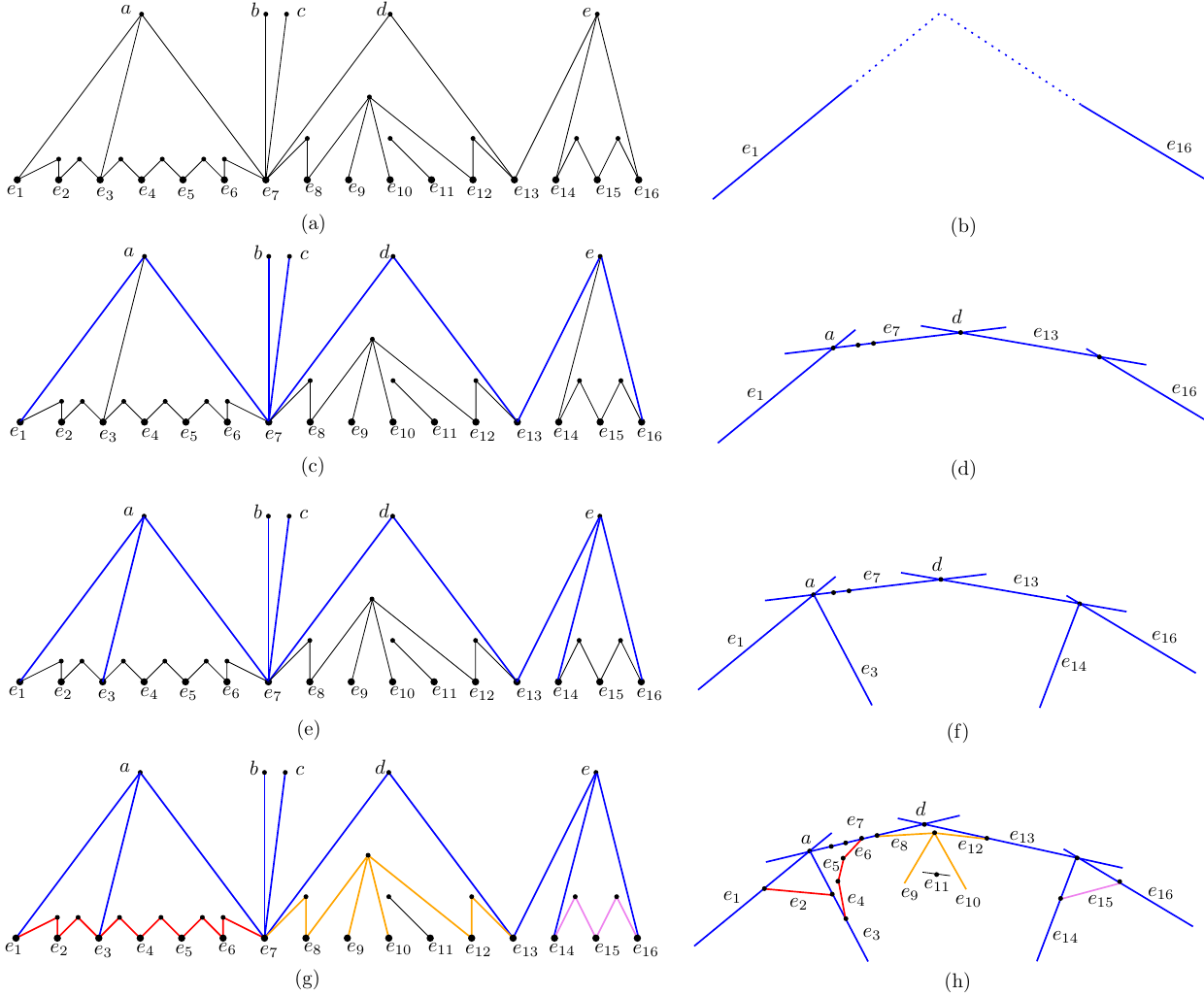}
    \caption{Illustration for \cref{lem:onesided}. Construction of a crossing-free strict segment representation in \cref{lem:onesided}}
    \label{fig:onesided}
\end{figure}
\begin{proof}
    Since $G$ is a bipartite graph, the girth (i.e., the length of a shortest cycle) of $G$ is at least $4$. If $G$ is of girth 4 then there are two hyperedges $e,e'$ in $H$ that contain at least two vertices in common. Therefore, $H$ cannot have a strict line segment representation. We can check whether the girth of $G$ is 4 in linear  time~\citep{chang2013computing}.
    If the girth of $G$ is at least 6, then we can construct a crossing-free strict segment representation as described below.

    Let $\Gamma$ be an embedding of $G$ where all the vertices of $E$ are on the outerface. %
    The embedding $\Gamma$ can be used to construct a straight-line planar drawing $\Gamma'$ where each vertex in $E$ lies on a horizontal line $L$, each vertex in $V$ appears strictly above $L$ with the $x$-coordinate lying between its leftmost and rightmost neighbours, and each edge in $E'$ drawn with a straight line segment between the corresponding endpoints (\cref{fig:onesided}(a)). Such a drawing can be computed in polynomial time~\citep{kaufmann2002embedding}.

    We repeatedly find the set of vertices $S$ in $\Gamma'$ (i.e., vertices and hyperedges of $H$) that lie on the unbounded region above $L$ and create points and segments to represent them. For example, in \cref{fig:onesided}(a), the selected set is  $S=\{e_1,a,e_7,b,c,d,e_{13},e_{16}\}$. Let $S_i$ be the set selected at the $i$th iteration.

    Let $\ell_s$ be the line determined by a line segment $s$. We will use a drawing invariant that the leftmost and rightmost hyperedges  $e,e'$ of $S_i$ are already drawn with distinct line segments $s,s'$ and the rest of the hyperedges of $S_i$ would be drawn using a convex chain in the region bounded by the lines $\ell_s, \ell_{s'}$.

    For $S_1$, we draw the leftmost and rightmost hyperedges $e,e'$ using segments $s,s'$ of positive and negative slopes, respectively, so that their corresponding lines form a convex chain (\cref{fig:onesided}(b)). If there is a vertex in $S$ that is incident to both $e,e'$, then we ensure that $s$ and $s'$ intersect, otherwise, we keep them disjoint. Let $e_1(=e),\ldots,e_t(=e')$ be the hyperedges of $H$ in $S$ which are in this left to right order on $L$. We draw $e_2,\ldots,e_{t-1}$ using a convex polygonal chain between $e,e'$ and place the vertices of $S$ on the chain to realize the necessary adjacencies.   \cref{fig:onesided}(c)--(d) illustrate the scenario for $S_1$. For each vertex $v$ on $S$, we now draw the hyperedges adjacent to it by splitting the angle created by leftmost and rightmost hyperedges, e.g., see vertex $a$ in \cref{fig:onesided}(e)--(f). We now continue with $S_i$, $i\ge 2$. Note that some hyperedges of $S_i$ are already drawn and we can use them to find maximal sequences in $S_i$ where the leftmost and rightmost hyperedges are already drawn. For example, \cref{fig:onesided}(g) shows the maximal sequences for $S_2$ in red, green, orange and purple. For each maximal sequence $e_p,\ldots,e_q$, where $1\le p,q \le |E|$, we draw the sequence with a convex chain between the segments representing $e_p$ and $e_q$. However, we also need to ensure the drawing remains crossing-free. If $S_{i-1}$ contains a vertex that is adjacent to both  $e_p$ and $e_q$, then $S_i$ cannot contain a vertex adjacent to both $e_p$ and $e_q$ (otherwise, the girth of $G$ would be 4). Therefore, we can draw $e_p,\ldots,e_q$ such that the drawing remains crossing-free. \cref{fig:onesided}(h) illustrates the line segment representation for the hypergraph of \cref{fig:onesided}(g).
\end{proof}

The representation that we construct in Lemma~\ref{lem:onesided} can be seen as a contact system of segments, and hence this identifies a class of graphs for which the existence of segment contact representation can be tested in polynomial time without using the conditions of~\citet{de2007representation,DBLP:journals/algorithmica/FraysseixM07} that check some properties over all subsystems of at least two paths.

\section{Beyond 0-bend representations}\label{section:beyondzerobend}
Researchers have examined various geometric representations of rank-$3$ max-degree-$3$ hypergraphs~\citep{glynnRepresentationConfigurationsProjective2000,kocay1999application}. In this section, we discuss rank-$3$ max-degree-$3$ hypergraphs and whether they have line representations with few bends.
Indeed, the PhD thesis of \citet{steinitz1894konstruction} claims that every $3$-uniform $3$-regular hypergraph can be represented with one line per hyperedge, except maybe one hyperedge (which could be represented with one bend).
However, more careful consideration shows that this is indeed not true, as already pointed out by \citet{DBLP:books/ph/Gruenbaum09}. %

We show a construction that has at least two hyperedges that must have a bend, and generalize this construction.
For this, we define, for $t\in \mathbb{N}_0$, a strict $t$-bend line representation for a hypergraph as follows. In the original definition of line representations, we replace lines by what we call \emph{infinite polygonal chains}. An infinite polygonal chain consists of a (possibly empty) polygonal chain and two rays, one ending at the first point of the chain, one ending at the last point (see \cref{fig:polychain}).
\begin{figure}[t]
    \centering
    \includegraphics[width=.9\textwidth]{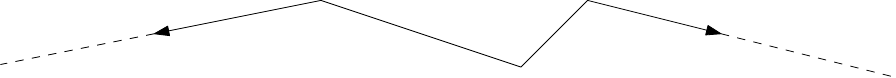}
    \caption{An infinite polygonal chain with $3$ bends extending infinitely to the left and right.}
    \label{fig:polychain}
\end{figure}
Essentially, an infinite polygonal chain is a polyline whose first and last segment is extended infinitely.
Further, two distinct infinite polygonal chains $\beta(e),\beta(e')$ for $e,e'\in E$ must not share a line segment nor a bend point.
Lastly, we require that the total number of all bends in $\beta(E)$ is exactly $t$, where $\beta(E)=\{\beta(e)\mid e\in E\}$. %

Consider the connected 3-uniform 3-regular hypergraph $H$ defined as follows.
Let $H_P^1$ and $H_P^2$ be two Pappus gadgets with vertices $p_1,\dots,p_9$ and $p_1',\dots,p_9'$, respectively.
The hypergraph $H$ is the union of $H_P^1,H_P^2$, and the hyperedges $\{p_7',p_8,p_9\}$ and $\{p_7,p_8',p_9'\}$. We have the following.
\begin{lemma}
    There is no $t$-bend representation for $H$ with $t<2$.
\end{lemma}
\begin{proof}
    If every hyperedge in the subhypergraph $H_P^2$ is represented without a bend, then $\beta(\{p_7',p_8,p_9\})$ must pass through $p_7$ due to \cref{thm:pappusthm}. Thus, at least one hyperedge of $H_P^1$ must be represented with at least one bend.
    Applying the same argument to $H_P^2$, we see that we require at least two bends.
\end{proof}

\begin{figure}[t]
    \centering
    \includegraphics[width=\linewidth]{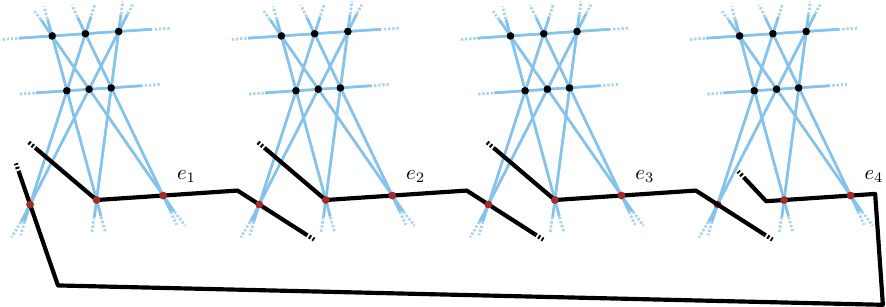}
    \caption{A 3-uniform 3-regular hypergraph requiring at least 4 bends in a strict line representation.}
    \label{fig:manybends}
\end{figure}
The above construction can be generalized so that there is no strict $t$-bend line representation for any $t<x$ for an arbitrary $x\in\mathbb{N}$. Instead of one copy of $H_P^1$, we add $x-1$ copies $H_P^2,H_P^3,\dots,H_P^x$. See also \cref{fig:manybends}.
Further, we add $x$ hyperedges $e_1,e_2,\dots,e_x$ such that $e_i$ contains $p_8,p_9$ of $H_P^i$ and $p_7$ of $H_P^{(i \mod x) + 1}$.
\begin{theorem}
    For any $x\in \mathbb{N}$ there exists a connected rank-3 hypergraph $H$ such that there exists no strict $t$-bend representation for $H$ for any $t<x$.
\end{theorem}
Similar constructions exist (at least for $x=2$) \citep{DBLP:books/ph/Gruenbaum09}, the generalization has not been stated explicitly.

\section{Conclusions and open problems}
We studied representations of hypergraphs by arrangements of lines/segments. While we answered several complexity questions, many open problems remain:
\begin{itemize}
    \item What is the time complexity of deciding whether a hypergraph has a crossing-free line representation?
    \item We have almost  characterized the complexity of deciding whether a  max-degree-2 hypergraph has a segment representation (hardness for rank at least five and a polynomial algorithm for rank at most three). What about max-degree-2 rank-4 hypergraphs?
    \item As most straight-line variants are hard, it is a natural question to explore the space where curves can have bends.
\end{itemize}

\bibliographystyle{abbrvnat}
\bibliography{literaturearxiv}

\begin{thebibliography}{33}
\providecommand{\natexlab}[1]{#1}
\providecommand{\url}[1]{\texttt{#1}}
\expandafter\ifx\csname urlstyle\endcsname\relax
  \providecommand{\doi}[1]{doi: #1}\else
  \providecommand{\doi}{doi: \begingroup \urlstyle{rm}\Url}\fi

\bibitem[Alper et~al.(2011)Alper, Riche, Ramos, and Czerwinski]{ahrc-dslnvt-11}
B.~Alper, N.~H. Riche, G.~A. Ramos, and M.~Czerwinski.
\newblock Design study of linesets, a novel set visualization technique.
\newblock \emph{{IEEE} Trans. Vis. Comput. Graph.}, 17\penalty0 (12):\penalty0
  2259--2267, 2011.
\newblock \doi{10.1109/TVCG.2011.186}.

\bibitem[Alsallakh et~al.(2016)Alsallakh, Micallef, Aigner, Hauser, Miksch, and
  Rodgers]{amahmr-sv-16}
B.~Alsallakh, L.~Micallef, W.~Aigner, H.~Hauser, S.~Miksch, and P.~J. Rodgers.
\newblock The state-of-the-art of set visualization.
\newblock \emph{Comput. Graph. Forum}, 35\penalty0 (1):\penalty0 234--260,
  2016.
\newblock \doi{10.1111/CGF.12722}.

\bibitem[Bertschinger et~al.(2023)Bertschinger, Maalouly, Kleist, Miltzow, and
  Weber]{DBLP:conf/gd/BertschingerMKMW23}
D.~Bertschinger, N.~E. Maalouly, L.~Kleist, T.~Miltzow, and S.~Weber.
\newblock The complexity of recognizing geometric hypergraphs.
\newblock In M.~A. Bekos and M.~Chimani, editors, \emph{Proc. Graph Drawing
  (GD'23)}, volume 14465 of \emph{LNCS}, pages 163--179. Springer, 2023.
\newblock \doi{10.1007/978-3-031-49272-3_12}.

\bibitem[Brandes et~al.(2012)Brandes, Cornelsen, Pampel, and
  Sallaberry]{bcps-psh-12}
U.~Brandes, S.~Cornelsen, B.~Pampel, and A.~Sallaberry.
\newblock Path-based supports for hypergraphs.
\newblock \emph{J. Discrete Algorithms}, 14:\penalty0 248--261, 2012.
\newblock \doi{10.1016/J.JDA.2011.12.009}.

\bibitem[Buchin et~al.(2009)Buchin, van Kreveld, Meijer, Speckmann, and
  Verbeek]{BuchinKMSV09}
K.~Buchin, M.~J. van Kreveld, H.~Meijer, B.~Speckmann, and K.~Verbeek.
\newblock On planar supports for hypergraphs.
\newblock In D.~Eppstein and E.~R. Gansner, editors, \emph{Graph Drawing
  (GD'09)}, volume 5849 of \emph{LNCS}, pages 345--356. Springer, 2009.
\newblock \doi{10.1007/978-3-642-11805-0_33}.

\bibitem[Buchin et~al.(2011)Buchin, van Kreveld, Meijer, Speckmann, and
  Verbeek]{bkmsv-psh-11a}
K.~Buchin, M.~J. van Kreveld, H.~Meijer, B.~Speckmann, and K.~Verbeek.
\newblock On planar supports for hypergraphs.
\newblock \emph{J. Graph Algorithms Appl.}, 15\penalty0 (4):\penalty0 533--549,
  2011.
\newblock \doi{10.7155/JGAA.00237}.

\bibitem[Castermans et~al.(2018)Castermans, van Garderen, Meulemans,
  Nöllenburg, and Yuan]{cgmny-spssh-18}
T.~Castermans, M.~van Garderen, W.~Meulemans, M.~Nöllenburg, and X.~Yuan.
\newblock Short plane supports for spatial hypergraphs.
\newblock In T.~Biedl and A.~Kerren, editors, \emph{Graph Drawing and Network
  Visualization (GD'18)}, volume 11282 of \emph{LNCS}, pages 53--66. Springer
  International Publishing, 2018.
\newblock \doi{10.1007/978-3-030-04414-5_4}.

\bibitem[Castermans et~al.(2019)Castermans, van Garderen, Meulemans,
  Nöllenburg, and Yuan]{cgmny-spssh-19}
T.~Castermans, M.~van Garderen, W.~Meulemans, M.~Nöllenburg, and X.~Yuan.
\newblock Short plane supports for spatial hypergraphs.
\newblock \emph{J. Graph Algorithms Appl.}, 23\penalty0 (3):\penalty0 463--498,
  2019.
\newblock \doi{10.7155/jgaa.00499}.

\bibitem[Chang and Lu(2013)]{chang2013computing}
H.~Chang and H.~Lu.
\newblock Computing the girth of a planar graph in linear time.
\newblock \emph{{SIAM} J. Comput.}, 42\penalty0 (3):\penalty0 1077--1094, 2013.
\newblock \doi{10.1137/110832033}.

\bibitem[Coxeter and Greitzer(1967)]{coxeter1967geometry}
H.~S.~M. Coxeter and S.~L. Greitzer.
\newblock \emph{Geometry revisited}, volume~19.
\newblock Mathematical Association of America, 1967.

\bibitem[de~Fraysseix and
  de~Mendez(2007)]{DBLP:journals/algorithmica/FraysseixM07}
H.~de~Fraysseix and P.~O. de~Mendez.
\newblock Representations by contact and intersection of segments.
\newblock \emph{Algorithmica}, 47\penalty0 (4):\penalty0 453--463, 2007.
\newblock \doi{10.1007/S00453-006-0157-X}.

\bibitem[de~Fraysseix et~al.(2007)de~Fraysseix, de~Mendez, and
  Rosenstiehl]{de2007representation}
H.~de~Fraysseix, P.~O. de~Mendez, and P.~Rosenstiehl.
\newblock Representation of planar hypergraphs by contacts of triangles.
\newblock In S.~Hong, T.~Nishizeki, and W.~Quan, editors, \emph{Proc. Graph
  Drawing (GD'07)}, volume 4875 of \emph{LNCS}, pages 125--136. Springer, 2007.
\newblock \doi{10.1007/978-3-540-77537-9_15}.

\bibitem[Eppstein(2009)]{eppstein2009testing}
D.~Eppstein.
\newblock Testing bipartiteness of geometric intersection graphs.
\newblock \emph{{ACM} Trans. Algorithms}, 5\penalty0 (2):\penalty0 15:1--15:35,
  2009.
\newblock \doi{10.1145/1497290.1497291}.

\bibitem[Firman and Spoerhase(2025)]{DBLP:journals/dmtcs/FirmanS25}
O.~Firman and J.~Spoerhase.
\newblock Hypergraph representation via axis-aligned point-subspace cover.
\newblock \emph{Discret. Math. Theor. Comput. Sci.}, 27\penalty0 (2), 2025.
\newblock \doi{10.46298/DMTCS.11676}.

\bibitem[Flowers(2015)]{flowers2015embeddings}
G.~Flowers.
\newblock \emph{Embeddings of configurations}.
\newblock PhD thesis, University of Victoria, 2015.

\bibitem[Frank et~al.(2021)Frank, Kaufmann, Kobourov, Mchedlidze, Pupyrev,
  Ueckerdt, and Wolff]{DBLP:conf/sofsem/FrankKKMPUW21}
F.~Frank, M.~Kaufmann, S.~G. Kobourov, T.~Mchedlidze, S.~Pupyrev, T.~Ueckerdt,
  and A.~Wolff.
\newblock Using the metro-map metaphor for drawing hypergraphs.
\newblock In B.~et~al., editor, \emph{Proc. Conference on Current Trends in
  Theory and Practice of Computer Science (SOFSEM'21)}, volume 12607 of
  \emph{LNCS}, pages 361--372. Springer, 2021.
\newblock \doi{10.1007/978-3-030-67731-2_26}.

\bibitem[Glynn(2000)]{glynnRepresentationConfigurationsProjective2000}
D.~G. Glynn.
\newblock On the representation of configurations in projective spaces.
\newblock \emph{Journal of Statistical Planning and Inference}, 86\penalty0
  (2):\penalty0 443--456, May 2000.
\newblock \doi{10.1016/S0378-3758(99)00124-X}.

\bibitem[Gon{\c{c}}alves(2009)]{DBLP:journals/ejc/Goncalves09}
D.~Gon{\c{c}}alves.
\newblock A planar linear hypergraph whose edges cannot be represented as
  straight line segments.
\newblock \emph{Eur. J. Comb.}, 30\penalty0 (1):\penalty0 280--282, 2009.
\newblock \doi{10.1016/J.EJC.2007.12.004}.

\bibitem[Gropp(1997)]{DBLP:journals/dm/Gropp97}
H.~Gropp.
\newblock Configurations and their realization.
\newblock \emph{Discret. Math.}, 174\penalty0 (1-3):\penalty0 137--151, 1997.
\newblock \doi{10.1016/S0012-365X(96)00327-5}.
\newblock URL \url{https://doi.org/10.1016/S0012-365X(96)00327-5}.

\bibitem[Gr\"{u}nbaum(2009)]{DBLP:books/ph/Gruenbaum09}
B.~Gr\"{u}nbaum.
\newblock \emph{Configurations of Points and Lines}.
\newblock American Mathematical Society, 2009.
\newblock ISBN 9780821843086.

\bibitem[Jacobsen et~al.(2021)Jacobsen, Wallinger, Kobourov, and
  N{\"{o}}llenburg]{jacobsen2020metrosets}
B.~Jacobsen, M.~Wallinger, S.~G. Kobourov, and M.~N{\"{o}}llenburg.
\newblock Metrosets: Visualizing sets as metro maps.
\newblock \emph{{IEEE} Trans. Vis. Comput. Graph.}, 27\penalty0 (2):\penalty0
  1257--1267, 2021.
\newblock \doi{10.1109/TVCG.2020.3030475}.

\bibitem[Johnson and Pollak(1987)]{jp-hpcdvd-87}
D.~S. Johnson and H.~O. Pollak.
\newblock Hypergraph planarity and the complexity of drawing venn diagrams.
\newblock \emph{J. Graph Theory}, 11\penalty0 (3):\penalty0 309--325, 1987.
\newblock \doi{10.1002/JGT.3190110306}.

\bibitem[Kaufmann and Wiese(2002)]{kaufmann2002embedding}
M.~Kaufmann and R.~Wiese.
\newblock Embedding vertices at points: Few bends suffice for planar graphs.
\newblock \emph{J. Graph Algorithms Appl.}, 6\penalty0 (1):\penalty0 115--129,
  2002.
\newblock \doi{10.7155/JGAA.00046}.

\bibitem[Kim et~al.(2024)Kim, de~Mesmay, and
  Miltzow]{DBLP:journals/corr/abs-2301-03221}
E.~J. Kim, A.~de~Mesmay, and T.~Miltzow.
\newblock Representing matroids over the reals is {$\exists
  \mathbb{R}$}-complete.
\newblock \emph{Discret. Math. Theor. Comput. Sci.}, 26\penalty0 (2), 2024.
\newblock \doi{10.46298/DMTCS.10810}.

\bibitem[Kocay and Szypowski(1999)]{kocay1999application}
W.~Kocay and R.~Szypowski.
\newblock The application of determining sets to projective configurations.
\newblock \emph{Ars Combinatoria}, 53:\penalty0 193--208, 1999.

\bibitem[Matousek(2014)]{DBLP:journals/corr/Matousek14}
J.~Matousek.
\newblock Intersection graphs of segments and {$\exists\mathbb{R}$}.
\newblock \emph{CoRR}, abs/1406.2636, 2014.

\bibitem[McConnell and Spinrad(1999)]{mcconnell1999modular}
R.~M. McConnell and J.~P. Spinrad.
\newblock Modular decomposition and transitive orientation.
\newblock \emph{Discret. Math.}, 201\penalty0 (1-3):\penalty0 189--241, 1999.
\newblock \doi{10.1016/S0012-365X(98)00319-7}.

\bibitem[Mäkinen(1990)]{m-dh-90}
E.~Mäkinen.
\newblock How to draw a hypergraph.
\newblock \emph{Int. J. Computer Math.}, 34\penalty0 (3--4):\penalty0 177--185,
  1990.
\newblock \doi{10.1080/00207169008803875}.

\bibitem[Schaefer(2009)]{s-csgtp-09}
M.~Schaefer.
\newblock Complexity of some geometric and topological problems.
\newblock In D.~Eppstein and E.~R. Gansner, editors, \emph{Proc. Graph Drawing
  (GD'09)}, volume 5849 of \emph{LNCS}, pages 334--344. Springer, 2009.
\newblock \doi{10.1007/978-3-642-11805-0_32}.

\bibitem[Shor(1990)]{shor1991stretchability}
P.~W. Shor.
\newblock Stretchability of pseudolines is {NP}-hard.
\newblock In P.~Gritzmann and B.~Sturmfels, editors, \emph{Proc. Applied
  Geometry And Discrete Mathematics (DIMACS'90)}, volume~4 of \emph{{DIMACS}
  Series in Discrete Mathematics and Theoretical Computer Science}, pages
  531--554. {DIMACS/AMS}, 1990.
\newblock \doi{10.1090/DIMACS/004/41}.

\bibitem[Steinitz(1894)]{steinitz1894konstruction}
E.~Steinitz.
\newblock \emph{{\"U}ber die Construction der Configurationen {\(n_3\)}}.
\newblock PhD thesis, Breslau, 1894.

\bibitem[Swaminathan and Wagner(1994)]{DBLP:journals/siamcomp/SwaminathanW94}
R.~Swaminathan and D.~K. Wagner.
\newblock On the consecutive-retrieval problem.
\newblock \emph{{SIAM} J. Comput.}, 23\penalty0 (2):\penalty0 398--414, 1994.
\newblock \doi{10.1137/S0097539792235487}.
\newblock URL \url{https://doi.org/10.1137/S0097539792235487}.

\bibitem[Wallinger et~al.(2021)Wallinger, Jacobsen, Kobourov, and
  N{\"{o}}llenburg]{wallinger2021readability}
M.~Wallinger, B.~Jacobsen, S.~G. Kobourov, and M.~N{\"{o}}llenburg.
\newblock On the readability of abstract set visualizations.
\newblock \emph{{IEEE} Trans. Vis. Comput. Graph.}, 27\penalty0 (6):\penalty0
  2821--2832, 2021.
\newblock \doi{10.1109/TVCG.2021.3074615}.

\end{thebibliography}
\label{sec:biblio}

\end{document}